%% file: main.tex
\RequirePackage{amsmath}

\documentclass[runningheads,dvipsnames]{llncs}
\usepackage{graphicx}
\usepackage{pgfplots}
\pgfplotsset{compat=1.18}
\usetikzlibrary{arrows.meta,shapes.geometric}

\usepackage{enumerate}
\usepackage{amssymb,amsfonts,latexsym}
\usepackage{hyperref}
\usepackage[capitalise,noabbrev]{cleveref}
\usepackage{xcolor}

\newcommand{\emptystr}{\varepsilon}

\newcommand{\floor}[1]{\lfloor #1 \rfloor}
\newcommand{\ceil}[1]{\lceil #1 \rceil}

\newcommand{\idtt}[1]{\ensuremath{\mathtt{#1}}}
\newcommand{\idsf}[1]{\ensuremath{\mathsf{#1}}}

\newcommand{\conj}[1]{[#1]}

\newcommand{\sqs}[1]{\mathit{SQ}_{#1}}
\newcommand{\factors}[1]{\mathit{Fac}_{#1}}
\newcommand{\lynfac}[1]{\mathit{Lyn}_{#1}}
\newcommand{\lynconj}[1]{\mathit{LC}_{#1}}
\newcommand{\rauzy}[1]{\Gamma_{#1}}

\newcommand{\cs}[1]{\mathit{CS}_{#1}}

\newcommand{\cir}[1]{C_{#1}} 

\newcommand{\kps}[2]{\mathit{KP}^{(#2)}_{#1}}

\newcommand{\nodes}[1]{\mathit{LT}(#1)}
\newcommand{\anchors}[1]{\mathcal{A}_{#1}}

\begin{document}
\title{
Tighter Upper Bounds on Distinct \texorpdfstring{$k$}{k}-Powers via Left Lyndon Trees
}

\author{Eitatsu Tomita\inst{1}\and
Peaker Guo\inst{2}\and 
Tomohiro I\inst{1}} 

\authorrunning{E. Tomita, P. Guo and T. I}

\institute{Kyushu Institute of Technology \\
\email{tomita.eitatsu846@mail.kyutech.jp},
\email{tomohiro@ai.kyutech.ac.jp}
\and
Institute of Science Tokyo
\email{peakerguo@gmail.com}
}

\maketitle
\begin{abstract}
A $k$-power is a word of the form $x^{k}$ for a non-empty word $x$; a $2$-power is called a square.
Li, Pachocki and Radoszewski [Electron. J. Comb., 2024] proved that the number of distinct $k$-powers in a word $w$ of length $n$ with $\sigma$ distinct letters is at most $\frac{n - \sigma}{k-1}$.
In this paper, we lower the bound to $\frac{1}{k-1}\left( n - \sigma - \log_{2} n + O(\log (k (\sigma + \log n))) \right)$.
For $k = 2$ this improves the bound $n - \Theta(\log n)$ of Brlek and Li [WORDS, 2023] by an additive $\sigma$ and makes the constant hidden in $\Theta(\log n)$ explicit.
Unlike previous work, which compares cycle lengths with the order of the Rauzy graph,
our proof classifies cycles by Lyndon words.
Given a long Lyndon word $z$ all of whose conjugates occur in $w$,
it extracts $\log_{2}|z|$ extra cycles that carry no $k$-power, one for each node of the left Lyndon tree of $z$.
\end{abstract}

\section{Introduction}\label{sec:intro}
For an integer $k \ge 2$, a \emph{$k$-power} is a word of the form $x^{k}$ for a non-empty word $x$;
a $2$-power is called a \emph{square}.
Among other formalisms of repetitions, squares and higher powers have been a central matter in combinatorics on words~\cite{1997Lothaire_CombinOnWords}
with the pioneering work on square-free words by Thue~\cite{1906Thue_UeberUnendZeich}.

Fraenkel and Simpson~\cite{1998FraenkelS_HowManySquarCanStrin} paid attention to
the size of the set $\sqs{w}$ of distinct squares in a finite word $w$ of length $n$.
While conjecturing $|\sqs{w}| < n$,
they proved $|\sqs{w}| < 2n$ by showing that at most two squares in $\sqs{w}$
share the same position for their rightmost occurrences.
Since then, their proof technique has been refined to get tighter upper bounds:
$2n - \Theta(\log n)$ by Ilie~\cite{2007Ilie_NoteOnNumberOfSquar},
$\frac{95}{48}n$ by Lam~\cite{2013Lam_NumberOfSquarInStrin},
$\frac{11}{6}n$ by Deza et al.~\cite{2015DezaFT_HowManyDoublSquarCan}, and
$1.5n$ by Thierry~\cite{2020Thierry_ProofThatWordOfLengt_X}.

Aside from this line of research, some variations of the conjecture have been raised:
Deza et al.~\cite{2011DezaFJ_DStepApproacForDistin_CPM} took the number $\sigma$ of distinct letters used in $w$ as a parameter and conjectured that $|\sqs{w}| \le n - \sigma$,
which was updated in~\cite{2014DezaF_DStepApproacToMaxim_DAM} to a slightly stronger conjecture
$|\sqs{w}| \le n - \sigma - \Theta(\log (n - \sigma))$.\footnote{
    In~\cite{2011DezaFJ_DStepApproacForDistin_CPM,2014DezaF_DStepApproacToMaxim_DAM}, only primitively rooted squares are considered.
}
Also, Jonoska et al.~\cite{2014JonoskaMS_StronSquarConjecOnBinar_SOFSEM} conjectured that
$|\sqs{w}| \le \frac{2\kappa-1}{2\kappa+2}n$ for a binary word $w$,
where $\kappa \le n/2$ is the number of occurrences of the less frequent letter in $w$.

The same question for higher powers has been studied as well.
Let $\kps{w}{k}$ denote the set of distinct $k$-powers occurring in $w$, so that $\kps{w}{2} = \sqs{w}$.
Kubica et al.~\cite{2013KubicaRRW_MaximNumberOfCubicSubwo} proved $|\kps{w}{3}| \le \frac{4}{5}n$ and that
the number of distinct factors of $w$ that are powers of exponent at least $3$ is at most $n-2$.
Li~\cite{2022Li_NumberOfKPowersInFinit} proved $|\kps{w}{k}| \le \frac{n-1}{k-2}$ for $k \ge 3$.

Recently, Brlek and Li~\cite{2025BrlekL_NumberOfSquarInFinit} proved that $|\sqs{w}| \le n - \sigma$,
breaking a barrier with a new approach based on Rauzy graphs~\cite{1982Rauzy_SuitesATermesDansUn}.
The Rauzy graph $\rauzy{w}(\ell)$ of order $\ell$ of $w$ is a directed graph such that
each vertex corresponds to a factor of length $\ell$ of $w$
and each arc corresponding to a factor of length $\ell+1$ of $w$
connects its prefix and suffix of length $\ell$.
Brlek and Li~\cite{2025BrlekL_NumberOfSquarInFinit} showed that
there is an injection from $\sqs{w}$ to an independent cycle set
of the graph $\bigcup_{\ell = 1}^{n} \rauzy{w}(\ell)$,
which is a collection of Rauzy graphs of orders from $1$ to $n$.
The result $|\sqs{w}| \le n - \sigma$ is obtained
because the maximum size of an independent cycle set of
$\bigcup_{\ell = 1}^{n} \rauzy{w}(\ell)$ is $n - \sigma$.
Furthermore, the techniques have been extended to show $|\sqs{w}| \le n - \Theta(\log n)$ in~\cite{2023BrlekL_NumberOfDistinSquarIn_WORDS}, and
to prove the conjecture of~\cite{2014JonoskaMS_StronSquarConjecOnBinar_SOFSEM} affirmatively in~\cite{2024Li_UpperBoundOfNumberOf}.
Most importantly for us, Li et al.~\cite{2024LiPR_NoteOnMaximNumberOf} generalized the approach to $k$-powers and proved
\begin{equation}\label{eq:lpr}
    |\kps{w}{k}| \le \frac{n - \sigma}{k-1}.
\end{equation}
They also gave, for every $k \ge 2$,
binary words of length $n$ containing $\frac{n}{k-1} - \Theta(\sqrt{n})$ distinct $k$-powers,
which generalizes the construction of $n - \Theta(\sqrt{n})$ distinct squares
by Fraenkel and Simpson~\cite{1998FraenkelS_HowManySquarCanStrin}.

The work~\cite{2024LiPR_NoteOnMaximNumberOf} determines the leading term $\frac{n}{k-1}$ of the maximum of $|\kps{w}{k}|$ over all words of length $n$.
However, it is expected that lower terms in~\cref{eq:lpr} can be lowered further.
Based on computational results of $|\sqs{w}|$ for binary words (listed as A248958 in OEIS~\cite{2026OEIS_A248958}),
Brlek and Li~\cite{2023BrlekL_NumberOfDistinSquarIn_WORDS} conjectured that
$|\sqs{w}| \le \ceil{n + 1 - \sqrt{n} - \log_{2} \left( \sqrt{n} \right)}$.
To pursue this quest, we prove the following bound, which is the main result of this paper.
\begin{theorem}\label{theo:main}
    For any integer $k \ge 2$ and any word $w$ of length $n$ using $\sigma$ distinct letters, 
    if $n \ge 2(\sigma + \log_{2} n)$, then
    \begin{align*}
        |\kps{w}{k}| &\le \frac{1}{k-1} \left( n - \sigma - \log_{2} \frac{n}{2(k-1)(\sigma + \log_{2} n)} \right)\\
                            & = \frac{1}{k-1} \left( n - \sigma - \log_{2} n + O(\log (k(\sigma + \log n))) \right).
    \end{align*}
\end{theorem}
For $k = 2$ and $\sigma \in O(1)$, \cref{theo:main} reads $|\sqs{w}| \le n - \sigma - \log_{2} n + \log_{2}\log_{2} n + O(1)$,
which improves the bound $n - \Theta(\log n)$ of~\cite{2023BrlekL_NumberOfDistinSquarIn_WORDS} by an additive $\sigma$
and makes the constant hidden in $\Theta(\log n)$ explicit.
It is also of the shape conjectured by Deza and Franek~\cite{2014DezaF_DStepApproacToMaxim_DAM}.
For general $k$, it improves~\cref{eq:lpr} by an additive $\frac{\log_{2} n}{k-1} - O\left(\frac{\log (k (\sigma + \log n))}{k}\right)$.
We remark that the saving cannot be made independent of $k$ (see~\cref{rem:obstruction}).

We also describe the connection between cycles and $k$-powers with a slightly different perspective from the previous work~\cite{2023BrlekL_NumberOfDistinSquarIn_WORDS,2024Li_NoteOnLieComplAnd,2024LiPR_NoteOnMaximNumberOf,2025BrlekL_NumberOfSquarInFinit}:
Instead of identifying interesting cycles by comparing their lengths with the order $\ell$ of Rauzy graphs, we classify the cycles based on Lyndon words.
We restate and prove the existing lemmas in this perspective to make this paper self-contained.
In~\cref{sec:anchored-circuits} we set up this framework and generalize the existing lemmas to $k$-powers.
The key new ingredient is that a single Lyndon word $z$ all of whose conjugates occur in $w$
yields $\log_{2}|z|$ cycles that can be reserved without carrying any $k$-power,
one for each node of the \emph{left Lyndon tree} of $z$.

Since $|\kps{w}{k}|$ is the output size of algorithms that report all distinct squares or $k$-powers of a word~\cite{GusfieldStoye2004,CrochemoreEtAl2014},
our result sharpens the bound on their output size, without giving a new algorithm.
Moreover, our proof relies on structures that are standard in string algorithms:
Lyndon roots, which underlie the computation of runs~\cite{BannaiEtAl2017Runs},
and left Lyndon trees, which can be built in linear time~\cite{2022BadkobehC_LinearConstOfLeftLyndon}.

\section{Preliminaries}\label{sec:prelim}
For two integers $b$ and $e$, let $[b..e]$ denote the set $\{ b, b+1, \dots, e \}$ of integers if $b \le e$, and otherwise, the empty set.
We also use $[b..]$ to denote the set $\{ b, b+1, \dots \}$ of integers that are larger than or equal to $b$.

A finite \emph{word} $w = w[1]w[2] \cdots w[|w|]$ is a sequence of letters, where $|w| \in [0..]$ denotes the length of $w$ and, for $i \in [1..|w|]$, $w[i]$ denotes the $i$-th letter of $w$.
The \emph{empty word} $\emptystr$ is the word of length $0$.
For integers $1 \le b \le e \le |w|$, the \emph{factor} of $w$ beginning at $b$ and ending at $e$ is denoted by $w[b..e] = w[b]w[b+1] \cdots w[e]$.
For convenience, let $w[b..e]$ with $b > e$ be the empty word.
A \emph{prefix} (resp. \emph{suffix}) of $w$ is a factor that matches $w[1..e]$ (resp. $w[b..|w|]$) for some $e \in [0..|w|]$ (resp. $b \in [1..|w|+1]$).
The \emph{concatenation} $xy$ of two words $x$ and $y$ is an associative binary operation such that $(xy)[1..|x|] = x$ and $(xy)[|x|+1..|xy|] = y$.

Let $\preceq$ denote an arbitrarily defined total order over letters,
which is also extended to represent the lexicographic order over words:
For any two words $x$ and $y$, $x \preceq y$ if and only if
either $x$ is a prefix of $y$ or there exists an integer $\ell \ge 0$ such that $x[1..\ell] = y[1..\ell]$ and $x[\ell+1] \prec y[\ell+1]$.

For any word $x$ and integer $m \ge 0$,
let $x^{m/|x|}$ denote the word of length $m$ such that $x^{m/|x|}[i] = x[j]$ for any $i \in [1..m]$ and $j \in [1..|x|]$ with $j \equiv i \mod |x|$.
If $w = x^{|w|/|x|}$ with $|x| \le |w|$, the integer $|x|$ is called a \emph{period} of $w$.
A word $w$ is called \emph{primitive} if it is not written as $w = x^{d}$ for some word $x$ and $d \in [2..]$.
For an integer $k \in [2..]$, a non-empty word is a \emph{$k$-power} if it is of the form $x^k$ for some word $x$ (here $x$ is not necessarily primitive).
In particular, a $2$-power is called a \emph{square}.
Throughout the paper, $k$ denotes a fixed integer with $k \ge 2$.

For any words $x$ and $y$, two words $xy$ and $yx$ are said to be \emph{conjugate}.
For any word $z$, the \emph{conjugacy class} $\conj{z}$ is defined as follows:
\begin{equation*}
    \conj{z} = \{ z[i..|z|]z[1..i-1] \mid i \in [1..|z|] \}
\end{equation*}
For any integer $m \ge 0$, we also define
\begin{equation*}
    \conj{z}_{m} = \{ x^{m/|x|} \mid x \in \conj{z} \}.
\end{equation*}
Remark that we allow $m$ to be smaller than $|z|$.
Intuitively, $\conj{z}_{m}$ is the set of factors of length $m$ in the infinite word $z^{\infty}$ that repeats $z$.

\begin{example}
    For $z = \idtt{aab}$, $\conj{z}_{0} = \{ \emptystr \}$, $\conj{z}_{1} = \{ \idtt{a}, \idtt{b} \}$, $\conj{z}_{2} = \{ \idtt{aa}, \idtt{ab}, \idtt{ba} \}$, $\conj{z}_{3} = \conj{z} = \{ \idtt{aab}, \idtt{aba}, \idtt{baa} \}$, $\conj{z}_{4} = \{ \idtt{aaba}, \idtt{abaa}, \idtt{baab} \}$, and $\conj{z}_{5} = \{ \idtt{aabaa}, \idtt{abaab}, \idtt{baaba} \}$.
\end{example}

We will mainly consider $\conj{z}$ and $\conj{z}_{m}$ for primitive words,
choosing \emph{Lyndon words} as a prominent representative of the conjugacy class for a primitive word.
\begin{definition}[\cite{1954Lyndon_BurnsSProbl,1997Lothaire_CombinOnWords}]\label{define:lyndon}
    A primitive word $z$ is called a \emph{Lyndon word}
    if $z$ is the lexicographically smallest word in $\conj{z}$, or equivalently,
    $z \prec z[i..|z|]$ holds for any $i \in [2..|z|]$.
\end{definition}

\begin{lemma}\label{lemma:lyndon_period}
    The smallest period of a Lyndon word $z$ is $|z|$.
\end{lemma}
\begin{proof}
    If $z$ has a period $p$ smaller than $|z|$, then its suffix $z[p+1..|z|]$ of length $|z| - p > 0$ is lexicographically smaller than $z$ because $z[p+1..|z|] = z[1..|z|-p] \prec z$, a contradiction.
    \qed
\end{proof}

The next lemma, which is implicitly used in~\cite{2025BrlekL_NumberOfSquarInFinit}, will be used repeatedly:
the prefixes of a Lyndon word $z$ are the lexicographically smallest words of $\conj{z}_{m}$.
\begin{lemma}\label{lemma:smallest}
    For any Lyndon word $z$ and integer $m \ge 0$, $z^{m/|z|}$ is a lexicographically smallest word in $\conj{z}_{m}$.
\end{lemma}
\begin{proof}
    Let $x \in \conj{z}$ with $x \neq z$.
    Since $z$ is the lexicographically smallest word in $\conj{z}$ and $|x| = |z|$,
    there is an integer $p \in [1..|z|]$ such that $z[1..p-1] = x[1..p-1]$ and $z[p] \prec x[p]$.
    If $m < p$, then $z^{m/|z|} = x^{m/|z|}$.
    Otherwise, $z^{m/|z|}[1..p-1] = x^{m/|z|}[1..p-1]$ and $z^{m/|z|}[p] \prec x^{m/|z|}[p]$, and hence $z^{m/|z|} \prec x^{m/|z|}$.
    \qed
\end{proof}

Since any $k$-power can be written as $x^{kr}$ for some primitive word $x$ and integer $r \ge 1$,
we associate it uniquely with the Lyndon word $z$ such that $x \in \conj{z}$.
We call such $z$ the \emph{Lyndon root} of the $k$-power.

For any word $w$, we define the following sets of words:
\begin{itemize}
    \item $\factors{w} = \{ x \mid w = uxv \}$: The set of factors of $w$.
    \item $\factors{w}(\ell) = \{ x \in \factors{w} \mid |x| = \ell \}$: The set of factors of length $\ell \ge 0$ in $w$.
    \item $\lynfac{w} = \{ z \in \factors{w} \mid z \text{~is a Lyndon word} \}$: The set of Lyndon factors of $w$.
    \item $\lynconj{w} = \{ z \in \lynfac{w} \mid \conj{z} \subseteq \factors{w} \}$: The set of Lyndon factors of $w$ all of whose conjugates occur in $w$.
    \item $\kps{w}{k} = \{ x \in \factors{w} \mid x \text{~is a $k$-power} \}$: The set of $k$-powers in $w$.
    \item $\kps{w}{k}(z) = \{ x^{kr} \in \kps{w}{k} \mid x \in \conj{z}, r \in [1..] \}$:
                The set of $k$-powers in $w$ with Lyndon root $z \in \lynfac{w}$.
\end{itemize}
We write $\sqs{w} = \kps{w}{k}$ and $\sqs{w}(z) = \kps{w}{k}(z)$ when $k = 2$ to emphasize that we are dealing with squares.

\begin{remark}\label{rem:partition}
    For any word $w$,
    \begin{equation*}
        \kps{w}{k} = \bigcup_{z \in \lynfac{w}} \kps{w}{k}(z)
        \quad\text{and}\quad
        |\kps{w}{k}| = \sum_{z \in \lynfac{w}} |\kps{w}{k}(z)|.
    \end{equation*}
\end{remark}

\begin{remark}\label{rem:lynconj}
    For any word $w$, $\lynconj{w}$ contains all the $\sigma$ distinct letters used in $w$, and hence $|\lynconj{w}| \ge \sigma$.
    The size of $\lynconj{w}$ is related to the \emph{Lie complexity} studied in~\cite{2024Li_NoteOnLieComplAnd}.
\end{remark}

A \emph{directed graph} consists of two distinct sets of vertices and arcs, where each arc connects two vertices with the direction from its \emph{initial vertex} to its \emph{terminal vertex}.\footnote{
    In our work, we need \emph{multigraphs} to deal with the Rauzy graph of order $0$,
    thus allow distinct arcs to have the same pair of initial and terminal vertices.
    We also allow a \emph{loop}, an arc whose initial vertex coincides with its terminal vertex.
}
A sequence $(u_1, u_2, \dots, u_m)$ of arcs is called a \emph{path} of length $m$ if the terminal vertex of $u_i$ coincides with the initial vertex of $u_{i+1}$ for any $i \in [1..m-1]$.
The path is called a \emph{circuit} if the terminal vertex of $u_m$ coincides with the initial vertex of $u_1$.
In paths and circuits, we consider moving from the initial vertex of $u_1$ to the terminal vertex of $u_m$ while following the direction of arcs in the sequence.
Note that we do not require the arcs of a path or a circuit to be pairwise distinct.
If we are allowed to traverse arcs also in the opposite direction, the graph is treated as an \emph{undirected} graph.
In the undirected case, the counterparts of paths and circuits are called \emph{chains} and \emph{cycles}, respectively (see~\cite{Berge1982TheorOfGraphAndIts,2025BrlekL_NumberOfSquarInFinit} for a more formal definition).
Note that a path (resp. circuit) is always a chain (resp. cycle) in this terminology.
A graph is called \emph{weakly connected} if, for every pair of distinct vertices, there is a chain going from one to the other.
A maximal weakly connected subgraph is called a \emph{component}.

\begin{definition}
    For a graph with the arc set $\{ u_1, u_2, \dots, u_{n_{\idsf{a}}} \}$,
    the \emph{cycle-vector} of a cycle $C$ is the vector $\mu(C) = (c_1, c_2, \dots, c_{n_{\idsf{a}}})$
    in the $n_{\idsf{a}}$-dimensional space such that $c_i = r_i - s_i$, where
    $r_i$ and $s_i$ are the numbers of times $u_i$ is traversed in one direction and respectively in the opposite direction in the cycle.
\end{definition}
\begin{definition}
    A set $\{ C_1, C_2, \dots, C_m \}$ of cycles is called \emph{independent} if their cycle-vectors are linearly independent, i.e., $\alpha_1 \mu(C_1) + \alpha_2 \mu(C_2) + \dots + \alpha_{m} \mu(C_m) = \mathbf{0}$ implies that $\alpha_1 = \alpha_2 = \dots = \alpha_m = 0$.
\end{definition}

\begin{theorem}[Theorem 2, Chapter 4 in \cite{Berge1982TheorOfGraphAndIts}]\label{theo:cyclomatic}
    For a multigraph with $n_{\idsf{v}}$ vertices, $n_{\idsf{a}}$ arcs and $n_{\idsf{c}}$ components,
    the \emph{cyclomatic number} of the graph, defined to be $n_{\idsf{a}} - n_{\idsf{v}} + n_{\idsf{c}}$, is equal to the maximum size of an independent cycle set.
\end{theorem}

\begin{definition}[Rauzy graph~\cite{1982Rauzy_SuitesATermesDansUn}]
    The \emph{Rauzy graph} $\rauzy{w}(\ell)$ of order $\ell \in [0..|w|]$ for a word $w$ is defined on the vertex set $\factors{w}(\ell)$ and the arc set $\factors{w}(\ell+1)$ such that each arc $u \in \factors{w}(\ell+1)$ begins at vertex $u[1..\ell] \in \factors{w}(\ell)$ and ends at vertex $u[2..|u|] \in \factors{w}(\ell)$.
    Further, let $\rauzy{w} = \bigcup_{\ell = 0}^{|w|} \rauzy{w}(\ell)$ be the graph that collects Rauzy graphs of all orders $\ell \in [0..|w|]$ in one graph with the vertex set $\factors{w}$, the arc set $\factors{w} \setminus \{ \emptystr \}$ and $|w|+1$ components.
\end{definition}
Note that every word $x \in \factors{w} \setminus \{ \emptystr \}$ appears as a vertex and also as an arc in $\rauzy{w}$, but they are distinguished from each other.
See \cref{fig:rauzy} for an illustration.

\begin{lemma}\label{lemma:rauzy_max_cycle}
    For any word $w$, the maximum size of an independent cycle set of $\rauzy{w}$ is $|w|$.
\end{lemma}
\begin{proof}
    Since $\rauzy{w}$ has $|\factors{w}|$ vertices, $|\factors{w} \setminus \{ \emptystr \}|$ arcs and $|w|+1$ components,
    its cyclomatic number is $|\factors{w} \setminus \{ \emptystr \}| - |\factors{w}| + |w| + 1 = |w|$,
    which is the maximum size of an independent cycle set of $\rauzy{w}$ by \cref{theo:cyclomatic}.
    \qed
\end{proof}

\begin{figure}[t]
    \centering
    \resizebox{0.90\linewidth}{!}{\input{figs/rauzy.tex}}
    \caption{
    Rauzy graphs $\rauzy{w}(\ell)$ for $\ell\in[0..8]$ and
    $w=\idtt{abbaabba}$.
    The red dashed arcs mark the anchors (lexicographically smallest arcs) of the seven circuits in $\bigcup_{z\in\lynconj{w}}\cs{w}(z)$ (defined in~\cref{sec:circuits_conj}), where
    $\lynconj{w}=\{\idtt{a},\idtt{b},\idtt{ab},\idtt{aabb}\}$.
    The blue dashed arc $\idtt{aab}$ anchors the additional circuit $(\idtt{aab},\idtt{abb},\idtt{bba},\idtt{baa})$,
    obtained from $\idtt{aab}\in\nodes{\idtt{aabb}}\setminus\lynconj{w}$ via \cref{lemma:minarc}.
    These eight circuits have distinct anchors and are therefore independent.
    For $k=2$, the squares are $\idtt{aa}$, $\idtt{bb}$, and $\idtt{abbaabba}$, with Lyndon roots $\idtt{a}$, $\idtt{b}$, and $\idtt{aabb}$, respectively.
    Here $\anchors{w}=\lynconj{w}\cup\{\idtt{aab}\}$, so
    $|\anchors{w}|=4+1=5$ (see \cref{def:anchorset} for the definition of $\anchors{w}$).
    Thus \cref{theo:master} gives $|\sqs{w}|\le 8-5=3$,
    with equality for the three squares listed.
    }
    \label{fig:rauzy}
\end{figure}

\section{Circuits Anchored at Their Smallest Arcs}\label{sec:anchored-circuits}
Brlek and Li~\cite{2025BrlekL_NumberOfSquarInFinit} showed that
there is an injection from $\sqs{w}$ to an independent cycle set in $\bigcup_{\ell = 1}^{|w|} \rauzy{w}(\ell)$,
and Li et al.~\cite{2024LiPR_NoteOnMaximNumberOf} generalized it to $k$-powers.
In this section we restate this construction in terms of Lyndon words and
generalize it to $k$-powers.

\subsection{Independence via Smallest Arcs}

We first isolate the mechanism that makes the cycles we use independent.
Since every arc of a circuit in $\rauzy{w}$ belongs to the same $\factors{w}(\ell+1)$ for some $\ell$,
all the arcs of a circuit are of the same length and can be compared lexicographically.

\begin{definition}\label{def:anchor}
    Let $C$ be a circuit in $\rauzy{w}$ and $x$ be an arc used by $C$.
    We say that $C$ is \emph{anchored at $x$} if 
    $x$ is a lexicographically smallest arc of $C$.
\end{definition}
Note that the coefficient of $x$ in the cycle-vector $\mu(C)$ of a circuit $C$ anchored at $x$ is positive
because every arc of a circuit is traversed in the same direction.

\begin{lemma}\label{lemma:anchor_independent}
    Let $\mathcal{C}$ be a set of circuits in $\rauzy{w}$ such that every $C \in \mathcal{C}$ is anchored at an arc $x_{C}$
    and the anchors are pairwise distinct.
    Then $\mathcal{C}$ is an independent cycle set of $\rauzy{w}$, and hence $|\mathcal{C}| \le |w|$.
\end{lemma}
\begin{proof}
    Let $C_1, C_2, \dots, C_m$ denote the circuits in $\mathcal{C}$ sorted so that $x_{C_1}, x_{C_2}, \dots, x_{C_m}$ is
    in strictly increasing order, first by their lengths and then lexicographically,
    which is possible because the anchors are pairwise distinct.
    Note that, for any $1 \le i < j \le m$, the anchor $x_{C_i}$ is not used by $C_j$:
    If $|x_{C_i}| \neq |x_{C_j}|$, the arcs of $C_i$ and $C_j$ are of different lengths;
    otherwise every arc of $C_j$ is lexicographically larger than or equal to $x_{C_j} \succ x_{C_i}$.
    We now show that
    \[
        \alpha_1 \mu(C_1) + \alpha_2 \mu(C_2) + \dots + \alpha_{m} \mu(C_m) = \mathbf{0} \Longrightarrow \alpha_1 = \alpha_2 = \dots = \alpha_m = 0.
    \]
    In order to have the equation $\sum_{i = 1}^{m} (\alpha_i \mu(C_i)) = \mathbf{0}$,
    we need $\alpha_1 = 0$, since otherwise,
    the component of the vector corresponding to $x_{C_1}$ cannot be zero.
    Applying the same argument recursively to the remaining equation $\sum_{i = 2}^{m} (\alpha_i \mu(C_i)) = \mathbf{0}$,
    we can see that all coefficients $\alpha_1, \alpha_2, \dots, \alpha_m$ must be zero.
    The upper bound $|\mathcal{C}| \le |w|$ follows from \cref{lemma:rauzy_max_cycle}.
    \qed
\end{proof}

\subsection{Circuits with Arc Set $\conj{z}_{m}$}\label{sec:circuits_conj}

We next recall two simple lemmas (\cref{lemma:conjm_circuit,lemma:conjm_reduce}),
which will be used to identify some circuits with the arc set $\conj{z}_{m}$ spreading across Rauzy graphs of different orders.

\begin{lemma}\label{lemma:conjm_reduce}
    For any word $w$ and $z \in \lynfac{w}$ with $\conj{z}_{m} \subseteq \factors{w}$ for some integer $m > 0$, $\conj{z}_{m-1} \subseteq \factors{w}$ holds.
\end{lemma}
\begin{proof}
    For any word $u \in \conj{z}_{m-1}$, there is a word $x \in \conj{z}$ such that $u = x^{(m-1)/|x|}$.
    Since $u$ is a prefix of $x^{m/|x|} \in \conj{z}_{m} \subseteq \factors{w}$,
    $u$ must be a factor of $w$.
    \qed
\end{proof}

\begin{lemma}\label{lemma:conjm_circuit}
    For any word $w$ and $z \in \lynfac{w}$ with $\conj{z}_{m} \subseteq \factors{w}$ for some integer $m > 0$, 
    let $x_i = z[i..|z|]z[1..i-1]$ for any $i \in [1..|z|]$.
    Then,
    \[
        \cir{z}(m) = \left( x_1^{m/|z|}, x_2^{m/|z|}, \dots, x_{|z|}^{m/|z|} \right)
    \]
    is a circuit in $\rauzy{w}(m-1)$ whose arc set is $\conj{z}_{m}$, and is anchored at $z^{m/|z|}$.
\end{lemma}
\begin{proof}
    For $i \in [1..|z|-1]$,
    the terminal vertex of the arc $x_{i}^{m/|z|}$ is $x_{i}^{m/|z|}[2..m] = x_{i+1}^{(m-1)/|z|} = x_{i+1}^{m/|z|}[1..m-1]$, which coincides with the initial vertex of the next arc $x_{i+1}^{m/|z|}$.
    Similarly, the terminal vertex of the arc $x_{|z|}^{m/|z|}$ is $x_{|z|}^{m/|z|}[2..m] = x_{1}^{(m-1)/|z|} = x_{1}^{m/|z|}[1..m-1]$, which coincides with the initial vertex of $x_{1}^{m/|z|}$.
    The circuit is anchored at $z^{m/|z|} = x_{1}^{m/|z|}$ by \cref{lemma:smallest}.
    \qed
\end{proof}

\begin{remark}
    If $m < |z|$, the circuit $\cir{z}(m)$ 
    may pass through a vertex twice and/or use an arc twice.
    On the other hand, for the case with $m \ge |z|$, it is possible to prove that the circuit passes through each vertex exactly once using Fine and Wilf's Periodicity Lemma~\cite{1965FineW_UniquenTheorForPeriodFunct},
    but we omit the proof because we do not use the property in this paper.
\end{remark}

\begin{definition}\label{def:cs}
    For any word $w$ and $z \in \lynfac{w}$,
    define
    \[
        \cs{w}(z) = \left\{ 
            \cir{z}(m) \; \mid \; m \geq |z| \text{ and } 
            \conj{z}_{m} \subseteq \factors{w}
        \right\}.
    \]
\end{definition}
For example, in \cref{fig:rauzy}, we have
$\cs{w}(\idtt{b}) = \{ \cir{\idtt{b}}(1), \cir{\idtt{b}}(2) \}$ and 
$\cs{w}(\idtt{aabb}) = \{ \cir{\idtt{aabb}}(4), \cir{\idtt{aabb}}(5)\}$.
It follows from \cref{lemma:conjm_reduce,lemma:conjm_circuit} that $\cs{w}(z) \neq \emptyset$ if and only if $z \in \lynconj{w}$.
In that case, $|\cs{w}(z)| = M_{z} - |z| + 1$, where $M_{z}$ is the largest integer $m$ with $\conj{z}_{m} \subseteq \factors{w}$.

We aim to assign $\kps{w}{k}(z)$ exclusively to $\cs{w}(z)$.\footnote{
    Note that $\cs{w}(z)$ is well-defined even for $z \in \lynfac{w}$ with $|\kps{w}{k}(z)| = 0$.
}
The next lemma shows that the sets of $\cs{w}(z)$ are disjoint for any pair of Lyndon factors
and their union constitutes an independent cycle set of $\rauzy{w}$:
\begin{lemma}\label{lemma:independent_cycle_set}
    For any word $w$, the set $\bigcup_{z \in \lynfac{w}} \cs{w}(z)$ contains $\sum_{z \in \lynfac{w}} |\cs{w}(z)|$ circuits, and is an independent cycle set of $\rauzy{w}$.
\end{lemma}
\begin{proof}
    By \cref{lemma:conjm_circuit}, every circuit $\cir{z}(m) \in \cs{w}(z)$ is anchored at $z^{m/|z|}$.
    We show that these anchors are pairwise distinct.
    For two distinct Lyndon factors $z_1$ and $z_2 \in \lynfac{w}$ with $|z_1| \le |z_2|$ and integer $m \ge |z_2|$,
    it must hold that $z_1^{m/|z_1|} \neq z_2^{m/|z_2|}$ because:
    If $|z_1| = |z_2|$, then it is clear that $z_1^{m/|z_1|} \neq z_2^{m/|z_2|}$ as $z_1 \neq z_2$.
    Otherwise, $z_1^{m/|z_1|} = z_2^{m/|z_2|}$ implies that $z_2$ has a period $|z_1| < |z_2|$, which contradicts \cref{lemma:lyndon_period}.
    Since two circuits in $\cs{w}(z)$ with different values of $m$ have anchors of different lengths,
    the anchors of the circuits in $\bigcup_{z \in \lynfac{w}} \cs{w}(z)$ are pairwise distinct,
    which leads to $|\bigcup_{z \in \lynfac{w}} \cs{w}(z)| = \sum_{z \in \lynfac{w}} |\cs{w}(z)|$
    and, by \cref{lemma:anchor_independent}, to the independence.
    \qed
\end{proof}

\begin{remark}\label{rem:anchor_lyndon}
    By \cref{lemma:lyndon_period}, the anchor $z^{m/|z|}$ of $\cir{z}(m)$ 
    is not a Lyndon word because it has a period $|z| < m$.
    Hence a Lyndon word $x$ is the anchor of a circuit in $\bigcup_{z \in \lynfac{w}} \cs{w}(z)$ if and only if $x \in \lynconj{w}$.
    We will exploit this fact in \cref{sec:lyndontree} to add more circuits to the independent cycle set.
\end{remark}

\subsection{Sizes of $\kps{w}{k}(z)$ and $\cs{w}(z)$}

The next lemma evaluates the sizes of $\kps{w}{k}(z)$ and $\cs{w}(z)$.
\begin{lemma}\label{lemma:kps_cs}
    For any word $w$ and $z \in \lynfac{w}$ with $|\kps{w}{k}(z)| \ge 1$,
    \begin{align*}
        |\kps{w}{k}(z)| &   =  |z|(r - 1) + s, \text{~and}\\
        |\cs{w}(z)|  & \ge |z|(kr - 2) + s + 1,
    \end{align*}
    where
    \[
    r = \max \{ i \in [1..] \mid \exists x^{ki} \in \kps{w}{k}(z) \} \quad \text{and} \quad
    s = \left| \{ x \in \conj{z} \mid x^{kr} \in \kps{w}{k}(z) \} \right|.
    \]
\end{lemma}
\begin{proof}
    Let $x_i = z[i..|z|]z[1..i-1]$ for any $i \in [1..|z|]$, and
    let $\{ q_j \}_{j = 1}^{s}$ be the set of integers with $0 < q_1 < q_2 < \cdots < q_s \le |z|$
    such that $x_{q_j}^{kr} \in \kps{w}{k}$ for any $j \in [1..s]$.
    For every pair $i \in [1..|z|]$ and $d \in [1..r-1]$, $x_{i}^{kd}$ forms a distinct $k$-power,
    which exists in $\kps{w}{k}(z)$ because $x_{i}^{kd}$ is a factor of $x_{q_1}^{kr}$;
    indeed $kd \le kr - k \le kr - 2$ and thus $x_{i}^{kd}$ occurs in $x_{q_1}^{kr}$ at position $i - q_1 + 1$ if $i \ge q_1$ and at position $i - q_1 + 1 + |z|$ otherwise.
    By the maximality of $r$, every $k$-power in $\kps{w}{k}(z)$ is of the form $x_{i}^{kd}$ with $d \in [1..r]$.
    Then, it is clear that $|\kps{w}{k}(z)| = |z|(r - 1) + s$ from the definition of $s$.

    In order to see $|\cs{w}(z)| \ge |z|(kr - 2) + s + 1$,
    let $g = \max \{ q_{j+1} - q_{j} \mid j \in [1..s] \}$, where $q_{s+1} = q_1 + |z|$.
    It is easy to see that $g \le |z| - s + 1$
    because $g$ is maximized when $q_{j+1} - q_{j} = 1$
    for every $j \in [1..s]$ except for the one that achieves maximum $g$.
    We now set $M = k|z|r - g + 1$ and claim that $\conj{z}_{M} \subseteq \factors{w}$.
    For any $i \in [1..|z|]$ with $i \in [q_j..q_{j+1}-1]$ for some $j \in [1..s]$,
    $x_i^{M/|z|}$ is a factor of $x_{q_j}^{kr}$
    because the suffix of $x_{q_j}^{kr}$ beginning at position $i - q_j + 1$ is a prefix of $x_i^{kr}$
    and its length $k|z|r - (i - q_j + 1) + 1$ is at least $k|z|r - g + 1 = M$ since $i - q_j + 1 \le g$,
    so that this suffix begins with $x_i^{M/|z|}$.
    Similarly, for any $i \in [1..q_1-1]$, $x_i^{M/|z|}$ is a factor of $x_{q_s}^{kr}$.
    Therefore, $\conj{z}_{M} \subseteq \factors{w}$. See \cref{fig:kps_cs} for an illustration.

    It follows from \cref{lemma:conjm_reduce} that $\conj{z}_{m} \subseteq \factors{w}$ for every $m \in [1..M]$.
    Hence, by \cref{lemma:conjm_circuit,def:cs},
    $\{ \cir{z}(m) \mid  m \in [|z|..M] \} \subseteq \cs{w}(z)$.
    Among these circuits, $\cs{w}(z)$ contains $M - |z| + 1 = k|z|r - |z| - g + 2 \ge |z|(kr - 2) + s + 1$ circuits,
    spreading across the Rauzy graphs of orders from $|z|-1$ to $M-1$.
    \qed
\end{proof}

\begin{figure}[t]
    \centering
    \input{figs/kps_cs.tex}
    \caption{Illustration for the proof of \cref{lemma:kps_cs}.
    Here consider the situation where $k = 3$, $|z| = 5$, $r = 2$, $s = 3$, $q_1 = 2$, $q_2 = 3$ and $q_3 = 5$.
    Since $q_{2} - q_{1} = 1, q_{3} - q_{2} = 2, q_{4} - q_{3} = q_{1} + |z| - q_{3} = 2$, we have $g = 2$.
    $\kps{w}{k}(z)$ contains $s$ $k$-powers of the form $x_{q_j}^{kr}$ for $j \in [1..s]$, which are aligned under $z^{\infty}$.
    We see that every factor of length $k|z|r - g + 1$ of $z^\infty$ (depicted with a red double-headed arrow) is guaranteed to be covered by $x_{q_j}^{kr}$ for some $j \in [1..s]$, and thus a factor of $w$.
    }
    \label{fig:kps_cs}
\end{figure}

The next lemma states that, on top of the $(k-1)|\kps{w}{k}(z)|$ circuits that we can afford for $\kps{w}{k}(z)$,
the set $\cs{w}(z)$ always contains a slack of at least one circuit,
which becomes larger when $\kps{w}{k}(z)$ contains $k$-powers with a short Lyndon root.
\begin{lemma}\label{lemma:slack}
    For any word $w$ and $z \in \lynconj{w}$,
    \begin{equation}\label{eq:slack}
    \begin{aligned}
        |\cs{w}(z)|-(k-1)|\kps{w}{k}(z)|
        &\ge
        \max\left\{
            1,\,
            \frac{|\cs{w}(z)|}{(k-1)|z|+1}
        \right\} \\
        &\ge
        \max\left\{
            1,\,
            \frac{|\kps{w}{k}(z)|}{|z|}
        \right\}.
    \end{aligned}
    \end{equation}
\end{lemma}
\begin{proof}
    Write $\ell = |z|$ and note that $(k-1)\ell + 1 \ge 2$.
    If $|\kps{w}{k}(z)| = 0$, the left-hand side is $|\cs{w}(z)| \ge 1$ because $z \in \lynconj{w}$,
    while the middle quantity is $\max\left(1, \frac{|\cs{w}(z)|}{(k-1)\ell+1}\right) \le |\cs{w}(z)|$ and the right one is $1$.
    Assume $|\kps{w}{k}(z)| \ge 1$ from now on and let $r$ and $s$ be the variables in \cref{lemma:kps_cs},
    so that $|\kps{w}{k}(z)| = \ell(r-1)+s$ with $1 \le s \le \ell$.

    We first prove the inequality
    \begin{equation}\label{eq:ratio}
        \ell |\cs{w}(z)| \ge \left( (k-1)\ell + 1 \right) |\kps{w}{k}(z)|
    \end{equation}
    which is equivalent to each of the second and the third bounds of~\cref{eq:slack}
    as far as the terms other than $1$ are concerned.
    Using $|\cs{w}(z)| \ge \ell(kr-2)+s+1$ from \cref{lemma:kps_cs} and then $s \le \ell$ with $k \ge 2$,
    \begin{align*}
        \ell |\cs{w}(z)| &- \left( (k-1)\ell + 1 \right)|\kps{w}{k}(z)|\\
        &\ge \ell\left( \ell(kr-2)+s+1 \right) - \left( (k-1)\ell+1 \right)\left( \ell(r-1)+s \right)\\
        &= \ell^{2}(k+r-3) + \ell(2-r) - s\left( (k-2)\ell + 1 \right)\\
        &\ge \ell^{2}(k+r-3) + \ell(2-r) - \ell\left( (k-2)\ell + 1 \right)\\
        &= \ell(\ell-1)(r-1) \ \ge\ 0 .
    \end{align*}
    It remains to see that the left-hand side of~\cref{eq:slack} is at least $1$:
    since $s \le \ell$ and $k \ge 2$,
    \begin{align*}
        |\cs{w}(z)| - (k-1)|\kps{w}{k}(z)|
        &\ge \ell(kr-2) + s + 1 - (k-1)\left(\ell(r-1) + s\right)\\
        &= \ell(k+r-3) + 1 - (k-2)s \\
        & \ge \ell(r-1) + 1 .
        \tag*{\qed}
    \end{align*}
\end{proof}

For $k = 2$, the first bound of~\cref{eq:slack} states that the number of circuits reserved for $\sqs{w}(z)$
exceeds the number of squares in $\sqs{w}(z)$, which is the extension by primitive circuits used in~\cite{2023BrlekL_NumberOfDistinSquarIn_WORDS}.

\begin{remark}\label{rem:direct}
    The slack of~\cref{eq:slack} already yields a saving by itself:
    Combining it with the fact that
    $\rauzy{w}(m-1)$ contains the circuit $\cir{z}(m)$ for every $m \in [1..M]$ whenever $\conj{z}_{M} \subseteq \factors{w}$
    (\cref{lemma:conjm_reduce,lemma:conjm_circuit}),
    which either shares its anchor $z^{m/|z|}$ with a circuit in $\cs{w}(y)$ for some $y$ with $|y| \le m$
    or can be added to the independent cycle set by~\cref{lemma:anchor_independent},
    one obtains
    $|\kps{w}{k}| \le \frac{n}{k-1} - \Theta\left(\frac{\log n}{(k-1)^{2}}\right)$,
    which for $k = 2$ is the bound $n - \Theta(\log n)$ of~\cite{2023BrlekL_NumberOfDistinSquarIn_WORDS},
    although it is proved there in a different way.
    The saving obtained in this way decays quadratically in $k$
    because a circuit in $\rauzy{w}(m-1)$ gains only $\frac{1}{k-1} - \frac{m}{(k-1)m+1} = \frac{1}{(k-1)((k-1)m+1)}$
    over the trivial ratio $\frac{1}{k-1}$ of $k$-powers per circuit, in which $k$ appears twice.
    In contrast, the next section reserves a $k$-independent number of circuits that carry no $k$-power.
\end{remark}

\section{More Circuits from Left Lyndon Trees}\label{sec:lyndontree}

Throughout this section, let $z$ be a Lyndon factor in $\lynconj{w}$, i.e., $\conj{z} \subseteq \factors{w}$.
By \cref{lemma:conjm_reduce}, $\conj{z}_{m} \subseteq \factors{w}$ holds for any $m \in [1..|z|]$;
recall that $\conj{z}_{m}$ is then exactly the set of factors of length $m$ of $z^{\infty}$.
The circuits of $\cs{w}(z)$ are anchored at the words $z^{m/|z|}$ with $m \ge |z|$,
and \cref{rem:anchor_lyndon} tells us that a Lyndon word that is not in $\lynconj{w}$ is never used as an anchor
of a circuit in $\bigcup_{z \in \lynfac{w}} \cs{w}(z)$.
The following lemma produces circuits anchored at such words.

\begin{lemma}\label{lemma:minarc}
    Let $w$ be a word, $z \in \lynconj{w}$ and $m \in [1..|z|]$.
    For every $x \in \conj{z}_{m}$ with $x[1..m-1] = z[1..m-1]$,
    the Rauzy graph $\rauzy{w}(m-1)$ contains a circuit anchored at $x$.
\end{lemma}
\begin{proof}
    Let $u = z[1..m-1]$ and, for any $i \in [1..]$, let $f_i = z^{\infty}[i..i+m-1]$ and $v_i = z^{\infty}[i..i+m-2]$,
    where $z^{\infty}$ denotes the infinite word that repeats $z$.
    Since $m \le |z|$, we have $\conj{z}_{m} = \{ f_i \mid i \in [1..|z|] \}$ and $\conj{z}_{m-1} = \{ v_i \mid i \in [1..|z|] \}$,
    and both are subsets of $\factors{w}$ by \cref{lemma:conjm_reduce}.
    Hence each $f_i$ is an arc of $\rauzy{w}(m-1)$ that begins at $v_i$ and ends at $v_{i+1}$.
    By \cref{lemma:smallest}, $u = z^{(m-1)/|z|}$ is a lexicographically smallest word in $\conj{z}_{m-1}$, i.e., $u \preceq v_i$ for any $i$.

    Since $x \in \conj{z}_{m}$, there is a position $i_0 \in [1..|z|]$ with $f_{i_0} = x$,
    and then $v_{i_0} = x[1..m-1] = u$.
    Let $i_1$ be the smallest integer with $i_1 > i_0$ and $v_{i_1} = u$,
    which exists because $v_{i_0+|z|} = v_{i_0} = u$.
    We claim that the circuit $C = (f_{i_0}, f_{i_0+1}, \dots, f_{i_1-1})$ of $\rauzy{w}(m-1)$,
    which begins and ends at the vertex $u$, is anchored at $x$.
    Indeed, its first arc is $f_{i_0} = x$, and for any $i \in [i_0+1..i_1-1]$ we have $v_i \neq u$ and thus $u \prec v_i$,
    which implies $x = u \cdot x[m] \prec v_i \cdot f_i[m] = f_i$.
    \qed
\end{proof}

\begin{remark}
    \cref{lemma:conjm_circuit} is the special case $x = z^{m/|z|}$ of \cref{lemma:minarc}, in which $i_1 = i_0 + |z|$ and $C$ is the whole circuit with arc set $\conj{z}_{m}$.
\end{remark}

By \cref{lemma:minarc}, every word obtained by extending the lexicographically smallest word of $\conj{z}_{m-1}$ by one letter within $z^{\infty}$ is an anchor.
The next definition captures the Lyndon words among such words.

\begin{definition}[left standard factorization, left Lyndon tree~\cite{2022BadkobehC_LinearConstOfLeftLyndon}]\label{def:lefttree}
    Let $z$ be a Lyndon word with $|z| \ge 2$ and let $u$ be the longest proper prefix of $z$ that is a Lyndon word.
    The factorization $z = uv$ is called the \emph{left standard factorization} of $z$, in which $v$ is again a Lyndon word.
    The \emph{left Lyndon tree} of $z$ is the binary tree defined recursively as follows:
    for $|z| = 1$ it is a single node labeled by $z$;
    for $|z| \ge 2$ its root is labeled by $z$ and its left and right subtrees are the left Lyndon trees of $u$ and $v$, respectively.
    We denote by $\nodes{z}$ the set of labels of the nodes of the left Lyndon tree of $z$.
\end{definition}

\begin{figure}[t]
    \centering
    \begin{minipage}[c]{0.34\linewidth}
        \centering
        \resizebox{\linewidth}{!}{\input{figs/left_lyndon_tree.tex}}
    \end{minipage}\hfill
    \begin{minipage}[c]{0.64\linewidth}
        \caption{
        The left Lyndon tree of $z = \idtt{aabab}$, 
        whose left standard factorization is $\idtt{aab} \cdot \idtt{ab}$ and 
        $\nodes{z} = \{ \idtt{aabab}, \idtt{aab}, \idtt{ab}, \idtt{a}, \idtt{b} \}$.
        Its height is $3 \ge \log_{2} |z|$ and the labels along the leftmost root-to-leaf path, $\idtt{aabab}$, $\idtt{aab}$ and $\idtt{a}$, are pairwise distinct.
        Note that $x[1..|x|-1]$ is a prefix of $z$ for every $x \in \nodes{z}$, as stated in \cref{lemma:node_prefix}.
        }
        \label{fig:lefttree}
    \end{minipage}
\end{figure}

See \cref{fig:lefttree} for an example. 
The following is a classical step in Duval's algorithm~\cite{1983Duval_FactorizingWordsOverOrder}
for computing the Lyndon factorization of a word.
\begin{lemma}[\cite{1983Duval_FactorizingWordsOverOrder}]\label{lemma:duval}
    Let $y$ be a word with a period $\pi$ such that $y[1..\pi]$ is a Lyndon word,
    and let $c$ be a letter with $c \succ y[|y|+1-\pi]$.
    Then $yc$ is a Lyndon word.
\end{lemma}

\begin{lemma}\label{lemma:node_prefix}
    For any Lyndon word $z$ and any $x \in \nodes{z}$, it holds that $x[1..|x|-1] = z[1..|x|-1]$.
\end{lemma}
\begin{proof}
    We proceed by induction on $|z|$, where the case $|z| = 1$ is trivial.
    Let $z = uv$ be the left standard factorization of $z$.

    We first show that $v[1..|v|-1]$ is a prefix of $z$.
    Assume to the contrary and let $j \in [1..|v|-1]$ be the smallest integer with $v[j] \neq z[j]$,
    so that $v[1..j-1] = z[1..j-1]$.
    Since $v$ is a proper suffix of the Lyndon word $z$, we have $z \prec v$ by \cref{define:lyndon},
    and since the first mismatch between $z$ and $v$ occurs at position $j$, it holds that $z[j] \prec v[j]$.
    Let $N = |u| + j$ and $y = z[1..N-1] = u \cdot v[1..j-1] = u \cdot z[1..j-1]$.
    Since $z[1..j-1]$ is a prefix of $y$, the word $y$ has a period $|u|$,
    and $y[1..|u|] = u$ is a Lyndon word.
    Moreover $y[|y|+1-|u|] = z[N-|u|] = z[j] \prec v[j] = z[N]$.
    It thus follows from \cref{lemma:duval} that $z[1..N]$ is a Lyndon word.
    Since $|u| < N = |u| + j \le |u| + |v| - 1 < |z|$,
    it is a proper prefix of $z$ that is a Lyndon word longer than $u$, contradicting the choice of $u$.

    Now let $x \in \nodes{z}$.
    The claim is trivial if $x = z$.
    If $x \in \nodes{u}$, the induction hypothesis gives $x[1..|x|-1] = u[1..|x|-1] = z[1..|x|-1]$ because $u$ is a prefix of $z$.
    If $x \in \nodes{v}$, then $|x| \le |v|$ and the induction hypothesis gives $x[1..|x|-1] = v[1..|x|-1]$,
    which is equal to $z[1..|x|-1]$ by the first part because $|x| - 1 \le |v| - 1$.
    \qed
\end{proof}

\begin{lemma}\label{lemma:node_count}
    For any Lyndon word $z$, every $x \in \nodes{z}$ is a Lyndon factor of $z$ and
    $|\{ x \in \nodes{z} \mid |x| \ge 2 \}| \ge \log_{2}|z|$.
\end{lemma}
\begin{proof}
    Every label of the left Lyndon tree is a Lyndon word and a factor of $z$ by \cref{def:lefttree}.
    Since every internal node of the tree has two children and every leaf is labeled by a single letter,
    the tree has exactly $|z|$ leaves and hence its height $h$ is at least $\log_{2}|z|$.
    Along a path from the root to a leaf of depth $h$, the lengths of the labels are strictly decreasing,
    so these $h+1$ labels are pairwise distinct and exactly one of them is of length $1$.
    \qed
\end{proof}

Every $z \in \lynconj{w}$ supplies its own left Lyndon tree, and we collect all of them.
\begin{definition}\label{def:anchorset}
    For any word $w$, let
    \begin{equation*}
        \anchors{w} = \bigcup_{z \in \lynconj{w}} \nodes{z}
        \quad\text{and}\quad
        \ell^{*} = \max \{ |z| \mid z \in \lynconj{w} \}.
    \end{equation*}
\end{definition}
Since $z \in \nodes{z}$ holds for every $z \in \lynconj{w}$, we have $\lynconj{w} \subseteq \anchors{w}$.
Moreover, $\lynconj{w}$ contains the $\sigma$ distinct letters of $w$ by \cref{rem:lynconj},
while $\nodes{z}$ contains at least $\log_{2}|z|$ words of length at least $2$ by \cref{lemma:node_count}, and hence
\begin{equation}\label{eq:anchors_lb}
    |\anchors{w}| \ge \sigma + \log_{2} \ell^{*} .
\end{equation}

\begin{lemma}\label{lemma:node_anchor}
    For any word $w$,
    \begin{equation*}
        \sum_{z \in \lynfac{w}} |\cs{w}(z)| + |\anchors{w} \setminus \lynconj{w}| \le |w|.
    \end{equation*}
\end{lemma}
\begin{proof}
    Consider the set $\mathcal{C}$ of circuits consisting of
    all the circuits in $\bigcup_{z \in \lynfac{w}} \cs{w}(z)$ and,
    for each $x \in \anchors{w} \setminus \lynconj{w}$, a circuit anchored at $x$.
    The latter circuits exist by \cref{lemma:minarc}:
    choosing $z \in \lynconj{w}$ with $x \in \nodes{z}$,
    $x$ is a factor of $z$ with $|x| \le |z|$ by \cref{lemma:node_count}, hence $x \in \conj{z}_{|x|}$,
    and $x[1..|x|-1] = z[1..|x|-1]$ by \cref{lemma:node_prefix}.
    The anchors of the former circuits are pairwise distinct by \cref{lemma:independent_cycle_set},
    the anchors of the latter circuits are pairwise distinct because they are distinct words,
    and no word is an anchor of both groups by \cref{rem:anchor_lyndon}
    since every $x \in \anchors{w}$ is a Lyndon word and the latter ones satisfy $x \notin \lynconj{w}$.
    Therefore, \cref{lemma:anchor_independent} applies to $\mathcal{C}$ and gives $|\mathcal{C}| \le |w|$,
    which is the claimed inequality.
    \qed
\end{proof}

\begin{remark}\label{rem:special}
    Recall that a \emph{border} of a word $x$ is a non-empty proper prefix of $x$ that is also a suffix of $x$,
    and that $x$ is \emph{bordered} if it has a border and \emph{unbordered} otherwise.
    In particular, Lyndon words are unbordered:
    Since a Lyndon word $z$ is lexicographically smaller than any of its proper suffixes $u$, 
    $u$ cannot be a prefix of $z$; otherwise, we would have both $z \prec u$ and $u \prec z$, a contradiction.
    Every element of $\anchors{w}$ is a Lyndon word and hence unbordered,
    so each circuit reserved in the proof of~\cref{lemma:node_anchor} is \emph{special}
    in the sense of~\cite{2023BrlekL_NumberOfDistinSquarIn_WORDS},
    where a circuit is called special if its arc set contains an unbordered word.
    It is shown there that a special circuit is independent of every circuit all of whose arcs are bordered,
    which is the case for the circuits that carry powers;
    this is how the bound $|\sqs{w}| \le n - \Theta(\log n)$ is obtained,
    the special circuits being extracted from the unbordered factors of the root of a long square.
    That argument, however, leaves open whether two special circuits in the same Rauzy graph are independent,
    so only one special circuit per Rauzy graph is used in~\cite{2023BrlekL_NumberOfDistinSquarIn_WORDS}.
    \cref{lemma:anchor_independent} removes this restriction as soon as the anchors are distinct,
    and \cref{lemma:minarc} supplies a circuit anchored at each prescribed Lyndon arc,
    so that all the $|\anchors{w}|$ circuits are counted at once;
    the $\sigma$ circuits anchored at the single letters, for instance, all live in $\rauzy{w}(0)$.
\end{remark}

We are now ready to prove our main results.
The first inequality of the next theorem shows that, on top of the $(k-1)|\kps{w}{k}|$ circuits charged to the $k$-powers,
we can reserve one circuit for each word in $\anchors{w}$.
\begin{theorem}\label{theo:master}
    For any word $w$ of length $n$,
    \begin{equation*}
        (k-1)|\kps{w}{k}| + \max \left( |\anchors{w}|, \frac{|\kps{w}{k}|}{\ell^{*}} \right) \le n,
    \end{equation*}
    where $\anchors{w}$ and $\ell^{*}$ are as in \cref{def:anchorset}.
\end{theorem}
\begin{proof}
    Note that $\cs{w}(z) = \emptyset$ for any $z \in \lynfac{w} \setminus \lynconj{w}$,
    and that $|\kps{w}{k}(z)| = 0$ for such $z$ because a $k$-power with Lyndon root $z$ contains all the conjugates of $z$.
    Hence $|\kps{w}{k}| = \sum_{z \in \lynconj{w}} |\kps{w}{k}(z)|$ and $|z| \le \ell^{*}$ for every $z \in \lynconj{w}$ with $|\kps{w}{k}(z)| \ge 1$.

    Summing~\cref{eq:slack} over $z \in \lynconj{w}$ and applying \cref{lemma:node_anchor}, we obtain
    \begin{align*}
        n &\ge \sum_{z \in \lynconj{w}} |\cs{w}(z)| + |\anchors{w} \setminus \lynconj{w}|\\
            &\ge (k-1)|\kps{w}{k}| + \sum_{z \in \lynconj{w}} \max \left( 1, \frac{|\kps{w}{k}(z)|}{|z|} \right) + |\anchors{w} \setminus \lynconj{w}|.
    \end{align*}
    Bounding the last sum from below by $|\lynconj{w}|$ gives
    $n \ge (k-1)|\kps{w}{k}| + |\lynconj{w}| + |\anchors{w} \setminus \lynconj{w}| = (k-1)|\kps{w}{k}| + |\anchors{w}|$,
    where we used $\lynconj{w} \subseteq \anchors{w}$.
    Bounding it instead by $\sum_{z \in \lynconj{w}} \frac{|\kps{w}{k}(z)|}{|z|} \ge \frac{|\kps{w}{k}|}{\ell^{*}}$
    gives $n \ge (k-1)|\kps{w}{k}| + \frac{|\kps{w}{k}|}{\ell^{*}}$.
    \qed
\end{proof}

Taking $\ell^{*} \ge 1$ and $|\anchors{w}| \ge \sigma$ in \cref{theo:master} recovers the bound~\cref{eq:lpr} of~\cite{2024LiPR_NoteOnMaximNumberOf}:
\begin{corollary}\label{coro:n-sigma}
    For any word $w$ of length $n$ using $\sigma$ distinct letters, $|\kps{w}{k}| \le \frac{n - \sigma}{k-1}$.
\end{corollary}

The two lower bounds in \cref{theo:master} are traded off against each other by the unknown value $\ell^{*}$.
Using $|\anchors{w}| \ge \sigma + \log_{2}\ell^{*}$ from~\cref{eq:anchors_lb},
the following corollary is obtained by balancing them.
\begin{corollary}\label{coro:main}
    For any word $w$ of length $n$ using $\sigma$ distinct letters with $|\kps{w}{k}| \ge 1$,
    \begin{equation*}
        (k-1)|\kps{w}{k}| + \sigma + \log_{2}\frac{|\kps{w}{k}|}{\sigma + \log_{2} n} \le n.
    \end{equation*}
\end{corollary}
\begin{proof}
    Let $P = |\kps{w}{k}|$ and $G = n - (k-1)P$,
    so that $G \ge \max(\sigma + \log_{2}\ell^{*}, P/\ell^{*})$ by \cref{theo:master} and~\cref{eq:anchors_lb}.
    If $\ell^{*} \ge \frac{P}{\sigma + \log_{2} n}$, then $G \ge \sigma + \log_{2}\ell^{*} \ge \sigma + \log_{2}\frac{P}{\sigma+\log_{2} n}$.
    Otherwise $G \ge \frac{P}{\ell^{*}} > \sigma + \log_{2} n \ge \sigma + \log_{2}\frac{P}{\sigma+\log_{2} n}$,
    where the last inequality holds because $P \le n$ and $\sigma + \log_{2} n \ge 1$.
    \qed
\end{proof}

\begin{proof}[of \cref{theo:main}]
    Let $P = |\kps{w}{k}|$.
    If $P \le \frac{n}{2(k-1)}$, then the claimed bound follows from $P(k-1) \le \frac{n}{2} \le n - \sigma - \log_2 n \le n - \sigma - \log_2 \frac{n}{2(k-1)(\sigma + \log_2 n)}$,
    where the second inequality comes from the assumption $n \ge 2(\sigma + \log_2 n)$.
    Otherwise $P > \frac{n}{2(k-1)}$ and \cref{coro:main} gives
    \begin{equation*}
        (k-1)P \le n - \sigma - \log_{2}\frac{P}{\sigma + \log_{2} n} < n - \sigma - \log_{2}\frac{n}{2(k-1)(\sigma + \log_{2} n)}.  
        \tag*{\qed}
    \end{equation*}
\end{proof}

\begin{remark}\label{rem:obstruction}
    The saving of $\frac{\log_{2} n}{k-1}$ obtained in \cref{theo:main} cannot be made independent of $k$.
    For the unary word $w = \idtt{a}^{n}$ we have $\kps{w}{k} = \{ \idtt{a}^{kd} \mid d \in [1..\floor{n/k}] \}$, and hence
    \begin{equation*}
        \frac{n - \sigma}{k-1} - |\kps{w}{k}| = \frac{n-1}{k-1} - \Bigl\lfloor \frac{n}{k} \Bigr\rfloor = \Theta\left( \frac{n}{k^{2}} \right)
    \end{equation*}
    for $k \le n/2$.
    Therefore, any upper bound of the form $\frac{n-\sigma}{k-1} - g(n,k)$ must satisfy $g(n,k) \in O(n/k^{2})$,
    and a saving of $\Theta(\log n)$ with a constant independent of $k$ is impossible for $k \in \Omega(\sqrt{n/\log n})$.
    We remark that \cref{theo:master} is tight for $w = \idtt{a}^{n}$:
    since $\ell^{*} = 1$, its second lower bound reads $(k-1)|\kps{w}{k}| + |\kps{w}{k}| \le n$, that is, $|\kps{w}{k}| \le n/k$.
\end{remark}

\bibliography{refs}

\section*{AI Disclosure}
The authors used a generative AI system (Anthropic's Claude, Opus 5, September 2026)
to assist in drafting and copy-editing parts of the text,  developing candidate proofs of some of the lemmas, and  producing the figures.
In particular, the balancing argument involving a trade-off between $\sigma + \log_2 \ell^{*}$ and $\frac{|\kps{w}{k}|}{\ell^{*}}$ was proposed by AI.
The remaining core ideas, including the use of left Lyndon trees, originated with the authors.
All statements and proofs were verified by the authors, who take full responsibility for the content.

\end{document}

%% file: figs/rauzy.tex
\begin{tikzpicture}[
    x=1cm, y=1cm,
    font=\scriptsize,
    >={Latex[length=1.6mm, width=1.3mm]},
    vertex/.style={draw, ellipse, inner sep=1.4pt, minimum height=4mm},
    singlevertex/.style={vertex, minimum width=6mm},
    arc/.style={->, semithick},
    csanchor/.style={arc, red, dash pattern=on 2.4pt off 0.8pt},
    ltanchor/.style={arc, blue, dash pattern=on 2.4pt off 0.8pt},
    lbl/.style={inner sep=1.5pt, text=black},
    panelborder/.style={draw=black!40, line width=0.5pt},
    loopleft/.style={to path={.. controls +(150:10.5mm) and +(210:10.5mm) .. (\tikztotarget) \tikztonodes}},
    loopright/.style={to path={.. controls +(-30:10.5mm) and +(30:10.5mm) .. (\tikztotarget) \tikztonodes}},
]
\def\W{4.0}   
\def\H{0.45}  
\def\GA{2.2}  
\def\GB{1.5}  
\pgfmathsetmacro{\Ya}{-(\H+\GA)}            
\pgfmathsetmacro{\Yb}{-2*(\H+\GA)}          
\pgfmathsetmacro{\Yc}{-2*(\H+\GA)-\H-\GB}   
\draw[panelborder] (0,0) rectangle (3*\W,\Yc);
\draw[panelborder] (\W,0) -- (\W,\Yc)  (2*\W,0) -- (2*\W,\Yc);
\draw[panelborder] (0,\Ya) -- (3*\W,\Ya)  (0,\Yb) -- (3*\W,\Yb);
\draw[panelborder] (0,-\H) -- (3*\W,-\H)  (0,\Ya-\H) -- (3*\W,\Ya-\H)  (0,\Yb-\H) -- (3*\W,\Yb-\H);
\foreach \i in {0,1,2} {
    \node[font=\small] at (\i*\W+\W/2, -\H/2)         {$\rauzy{w}(\i)$};
    \node[font=\small] at (\i*\W+\W/2, \Ya-\H/2)      {$\rauzy{w}(\the\numexpr\i+3\relax)$};
    \node[font=\small] at (\i*\W+\W/2, \Yb-\H/2)      {$\rauzy{w}(\the\numexpr\i+6\relax)$};
}
\coordinate (c0) at (0.5*\W, -\H-\GA/2);
\coordinate (c1) at (1.5*\W, -\H-\GA/2);
\coordinate (c2) at (2.5*\W, -\H-\GA/2);
\coordinate (c3) at (0.5*\W, \Ya-\H-\GA/2);
\coordinate (c4) at (1.5*\W, \Ya-\H-\GA/2);
\coordinate (c5) at (2.5*\W, \Ya-\H-\GA/2);
\coordinate (c6) at (0.5*\W, \Yb-\H-\GB/2);
\coordinate (c7) at (1.5*\W, \Yb-\H-\GB/2);
\coordinate (c8) at (2.5*\W, \Yb-\H-\GB/2);

\begin{scope}[shift={(c0)}]
    \node[singlevertex] (e) at (0,0) {$\emptystr$};
    \draw[csanchor] (e) to[loopleft]  node[lbl, midway, left]  {\idtt{a}} (e);
    \draw[csanchor] (e) to[loopright] node[lbl, midway, right] {\idtt{b}} (e);
\end{scope}

\begin{scope}[shift={(c1)}]
    \node[singlevertex] (a) at (-0.6,0) {\idtt{a}};
    \node[singlevertex] (b) at ( 0.6,0) {\idtt{b}};
    \draw[csanchor] (a) to[loopleft]  node[lbl, midway, left]  {\idtt{aa}} (a);
    \draw[csanchor] (b) to[loopright] node[lbl, midway, right] {\idtt{bb}} (b);
    \draw[csanchor] (a) to[bend left=45] node[lbl, midway, above=1pt] {\idtt{ab}} (b);
    \draw[arc]      (b) to[bend left=45] node[lbl, midway, below=1pt] {\idtt{ba}} (a);
\end{scope}

\begin{scope}[shift={(c2)}]
    \node[vertex] (L) at (-1.15, 0)    {\idtt{aa}};
    \node[vertex] (T) at ( 0,    0.75) {\idtt{ab}};
    \node[vertex] (R) at ( 1.15, 0)    {\idtt{bb}};
    \node[vertex] (B) at ( 0,   -0.75) {\idtt{ba}};
    \draw[ltanchor] (L) to[bend left=20] node[lbl, midway, above left]  {\idtt{aab}} (T);
    \draw[arc]      (T) to[bend left=20] node[lbl, midway, above right] {\idtt{abb}} (R);
    \draw[arc]      (R) to[bend left=20] node[lbl, midway, below right] {\idtt{bba}} (B);
    \draw[arc]      (B) to[bend left=20] node[lbl, midway, below left]  {\idtt{baa}} (L);
\end{scope}

\begin{scope}[shift={(c3)}]
    \node[vertex] (L) at (-1.15, 0)    {\idtt{aab}};
    \node[vertex] (T) at ( 0,    0.75) {\idtt{abb}};
    \node[vertex] (R) at ( 1.15, 0)    {\idtt{bba}};
    \node[vertex] (B) at ( 0,   -0.75) {\idtt{baa}};
    \draw[csanchor] (L) to[bend left=20] node[lbl, midway, above left]  {\idtt{aabb}} (T);
    \draw[arc]      (T) to[bend left=20] node[lbl, midway, above right] {\idtt{abba}} (R);
    \draw[arc]      (R) to[bend left=20] node[lbl, midway, below right] {\idtt{bbaa}} (B);
    \draw[arc]      (B) to[bend left=20] node[lbl, midway, below left]  {\idtt{baab}} (L);
\end{scope}

\begin{scope}[shift={(c4)}]
    \node[vertex] (L) at (-1.15, 0)    {\idtt{aabb}};
    \node[vertex] (T) at ( 0,    0.75) {\idtt{abba}};
    \node[vertex] (R) at ( 1.15, 0)    {\idtt{bbaa}};
    \node[vertex] (B) at ( 0,   -0.75) {\idtt{baab}};
    \draw[csanchor] (L) to[bend left=20] node[lbl, midway, above left]  {\idtt{aabba}} (T);
    \draw[arc]      (T) to[bend left=20] node[lbl, midway, above right] {\idtt{abbaa}} (R);
    \draw[arc]      (R) to[bend left=20] node[lbl, midway, below right] {\idtt{bbaab}} (B);
    \draw[arc]      (B) to[bend left=20] node[lbl, midway, below left]  {\idtt{baabb}} (L);
\end{scope}

\begin{scope}[shift={(c5)}]
    \node[vertex] (L) at (-1.15, 0)    {\idtt{aabba}};
    \node[vertex] (T) at ( 0,    0.75) {\idtt{abbaa}};
    \node[vertex] (R) at ( 1.15, 0)    {\idtt{bbaab}};
    \node[vertex] (B) at ( 0,   -0.75) {\idtt{baabb}};
    \draw[arc] (T) to[bend left=20] node[lbl, midway, above right] {\idtt{abbaab}} (R);
    \draw[arc] (R) to[bend left=20] node[lbl, midway, below right] {\idtt{bbaabb}} (B);
    \draw[arc] (B) to[bend left=20] node[lbl, midway, below left]  {\idtt{baabba}} (L);
\end{scope}

\begin{scope}[shift={(c6)}]
    \node[vertex] (L) at (-1.05,-0.32) {\idtt{abbaab}};
    \node[vertex] (T) at ( 0,    0.32) {\idtt{bbaabb}};
    \node[vertex] (R) at ( 1.05,-0.32) {\idtt{baabba}};
    \draw[arc] (L) to[bend left=20] node[lbl, midway, above left]  {\idtt{abbaabb}} (T);
    \draw[arc] (T) to[bend left=20] node[lbl, midway, above right] {\idtt{bbaabba}} (R);
\end{scope}

\begin{scope}[shift={(c7)}]
    \node[vertex] (L) at (-0.95,-0.15) {\idtt{abbaabb}};
    \node[vertex] (R) at ( 0.95,-0.15) {\idtt{bbaabba}};
    \draw[arc] (L) to[bend left=35] node[lbl, midway, above] {\idtt{abbaabba}} (R);
\end{scope}

\begin{scope}[shift={(c8)}]
    \node[vertex] at (0,0) {\idtt{abbaabba}};
\end{scope}
\end{tikzpicture}

%% file: figs/kps_cs.tex
\begin{tikzpicture}[scale=0.78,x=3.6mm,y=1mm,font=\footnotesize]
  \def\zlen{5}      
  \def\reps{6}      
  \def\M{29}        
  \foreach \i in {0,...,6}{
    \draw (\zlen*\i,0) rectangle ++(\zlen,6);
    \node at (\zlen*\i+\zlen/2,3) {$z$};
  }
  \node at (36.4,3) {$\cdots$};
  \foreach \p in {1,2,3,4,5,6}{
    \draw[dashed,gray!80] (\p,7.5) -- (\p,-48);
  }
  \foreach \p/\lab in {1/1,2/2,4/3,6/4}{
    \draw[-latex,thick] (\p+0.5,11) -- (\p+0.5,6.4);
    \node[anchor=south] at (\p+0.5,11) {$q_{\lab}$};
  }
  \foreach \off/\j/\y in {1/1/-12, 2/2/-24, 4/3/-40}{
    \foreach \i in {0,...,5}{
      \draw (\off+\zlen*\i,\y) rectangle ++(\zlen,6);
      \node at (\off+\zlen*\i+\zlen/2,\y+3) {$x_{q_{\j}}$};
    }
  }
  \draw[latex-latex,red] (1,-14.5) -- (\M+1-1,-14.5);
  \draw[latex-latex,red] (2,-26.5) -- (\M+2-1,-26.5);
  \draw[latex-latex,red] (3,-30.0) -- (\M+3-1,-30.0);
  \draw[latex-latex,red] (4,-42.5) -- (\M+4-1,-42.5);
  \draw[latex-latex,red] (5,-46.0) -- (\M+5-1,-46.0);
  \draw[latex-latex] (2,-37) -- (4,-37);
  \node[anchor=south,inner sep=1pt] at (3,-36.6) {$g$};
\end{tikzpicture}

%% file: figs/left_lyndon_tree.tex
\begin{tikzpicture}[level distance=8mm,font=\footnotesize,
        level 1/.style={sibling distance=24mm},
        level 2/.style={sibling distance=13mm},
        level 3/.style={sibling distance=8mm}]
    \node {\idtt{aabab}}
        child { node {\idtt{aab}}
            child { node {\idtt{a}} }
            child { node {\idtt{ab}}
                child { node {\idtt{a}} }
                child { node {\idtt{b}} } } }
        child { node {\idtt{ab}}
            child { node {\idtt{a}} }
            child { node {\idtt{b}} } };
\end{tikzpicture}